\RequirePackage{fix-cm}
\documentclass[smallextended]{svjour3}       
\smartqed  
\usepackage{graphicx}
\usepackage{amsmath,amsfonts,bm,mathtools}
\usepackage{ulem}

\begin{document}

\title{Classic Lagrangian may not be applicable to the traveling salesman problem
}


\author{Michael X. Zhou
}


\institute{Michael X. Zhou \at
              School of Information Science and Engineering, Central South University, Changsha 410083, China.\\
}

\date{Received: date / Accepted: date}

\maketitle

\begin{abstract}
In this short note, the dual problem for the traveling salesman problem is constructed through the classic Lagrangian.
The existence of optimality conditions is expressed as a corresponding inverse problem. 
A general 4-cities instance is given, and the numerical experiment shows that the classic Lagrangian may not be applicable to the traveling salesman problem.
\end{abstract}

\section{Problem description}
Let $\mathcal{N} = \{1,2,\cdots,n\}$ be a set of cities, and the distance between city \textit{i} and \textit{j} is given by $d_{ij}$ (for symmetrical TSP, $d_{ij} = d_{ji}, d_{ii} = 0, \forall i,j \in \mathcal{N}$), then the TSP can be represented by the quadratic programming problem (primal problem) \cite{Hopfield}
\begin{eqnarray}
(\mathcal{P})~~ \min && f(\bm X) = \frac{1}{2}\sum_{i=1}^{n}\sum_{k=1}^{n}\sum_{j=1}^{n} x_{ij}d_{ik}(x_{k(j-1)}+x_{k(j+1)}) \nonumber \\
\mathrm{subject~to} && \sum_{j=1}^{n} x_{ij} = 1 \nonumber \\
&& \sum_{i=1}^{n} x_{ij} = 1 \nonumber \\
&& x_{ij} \in \{0,1\}
\end{eqnarray}
where $\bm X = [x_{11}, x_{21},\cdots, x_{n1}, x_{12}, x_{22}, \cdots, x_{n2},\cdots, x_{1n}, x_{2n}, \cdots, x_{nn}]^T \in \mathbb{R}^{n^2}$, and $x_{ij}$ is defined by
\begin{equation}
x_{ij} =
\left\{ \begin{aligned}
1 &~~\mathrm{if~city}~i~\mathrm{is~in~the}~j\mathrm{th~position} \\
0 &~~\mathrm{otherwise}
\end{aligned} \right.
\end{equation}
\indent Furthermore, due to the round trip of TSP, we have
 \begin{equation}
x_{i0} = x_{in}, x_{i1} = x_{i(n+1)}, \forall i,j \in \mathcal{N}
\end{equation}

\begin{theorem}
The TSP can be rewritten to the following vector type
\begin{eqnarray}\label{eq4}
(\mathcal{P}) ~~\min && f(\bm X) = \frac{1}{2}\bm X^T A \bm X \nonumber \\
\mathrm{s.t.} && C \bm X = \bm e \nonumber \\
&& D \bm X = \bm e \nonumber \\
&& \bm X \circ \bm X - \bm X = \bm 0
\end{eqnarray}
\end{theorem}
where, $\bm e$ is an appropriate vector with entries one, $\bm s \circ \bm t = (s_1t_1, \cdots, s_nt_n)^T$ ($\bm s, \bm t \in \mathbb{R}^n$) is the Hadamard product, and
\begin{eqnarray}\label{eq5}
A =&& \sum_{k=1}^{n} \Bigg\{
\mathrm{diag}\Big\{d_{1k}, \cdots, d_{nk}, \cdots, d_{1k}, \cdots, d_{nk}\Big\}
\Big(\bm e^T_{(n-1)n+k}+\bm e^T_{n+k}; \cdots; \nonumber  \\
&& \bm e^T_{(n-1)n+k}+\bm e^T_{n+k}; \bm e^T_{k}+\bm e^T_{2n+k}; \cdots; \bm e^T_{k}+\bm e^T_{2n+k}; \cdots; \bm e^T_{(n-2)n+k}+\bm e^T_{k};  \nonumber  \\
&& \cdots; \bm e^T_{(n-2)n+k}+\bm e^T_{k}\Big)
\Bigg\}, \nonumber \\
C = && \begin{pmatrix}
1 & \cdots & 1 & 0 & \cdots & 0 & 0 & \cdots & 0 \\
0 & \cdots & 0 & 1 & \cdots & 1 & 0 & \cdots & 0 \\
\vdots & \ddots & \vdots & \vdots & \ddots & \vdots & \vdots & \ddots & \vdots \\
0 & \cdots & 0 & 0 & \cdots & 0 & 1 & \cdots & 1 \\
\end{pmatrix},
\;
D = \begin{pmatrix}
1 & 0 & \cdots & 0 & \cdots & \cdots & 1 & 0 & \cdots & 0 \\
0 & 1 & \cdots & 0 & \cdots & \cdots & 0 & 1 & \cdots & 0 \\
\vdots & \vdots & \ddots & \vdots &  \vdots & \vdots & \vdots & \vdots & \ddots & \vdots \\
0 & 0 & \cdots & 1 & \cdots & \cdots & 0 & 1 & \cdots & 1 \\
\end{pmatrix}
\end{eqnarray}
here, $A \in \mathbb{R}^{n^2 \times n^2}$, $C,D \in \mathbb{R}^{n \times n^2}$ and $\bm {e}_i \in \mathbb{R}^{n^2}$ is an unit vector with nonzero at position $i$.

\begin{proof}
Let define the following notation
\begin{eqnarray}
U \oplus V \oplus W = \sum_{i,j} u_{ij} v_{ij} w_{ij}, U = \{u_{ij}\}^{n \times n}, V = \{v_{ij}\}^{n \times n}, W = \{w_{ij}\}^{n \times n}. \;\;\;
\end{eqnarray}
In the objective function, the first part can be rewritten as
\begin{eqnarray}
\sum_{i,j}x_{ij}d_{ik}x_{k(j-1)}
& = & \begin{pmatrix}
x_{11} & \cdots & x_{1n} \\
\vdots & \ddots & \vdots \\
x_{n1} & \cdots & x_{nn}
\end{pmatrix}
\oplus
\begin{pmatrix}
d_{1k} & \cdots & d_{1k} \\
\vdots & \ddots & \vdots \\
d_{nk} & \cdots & d_{nk}
\end{pmatrix}
\oplus
\begin{pmatrix}
x_{kn} & \cdots & x_{k(n-1)} \\
\vdots & \ddots & \vdots \\
x_{kn} & \cdots & x_{k(n-1)}
\end{pmatrix} \nonumber \\
& = &
\begin{pmatrix}
x_{11}\\
\vdots \\
x_{n1} \\
\vdots \\
x_{1n}\\
\vdots \\
x_{nn} \\
\end{pmatrix}^T
\begin{pmatrix}
d_{1k} &  & & & & &\\
       & \ddots & & & & &\\
       &   & d_{nk} & & &\\
       &   &  & \ddots & & &\\
       &  & & & d_{1k} & &\\
       & & & & & \ddots &\\
       &   &  & & &  & d_{nk}
\end{pmatrix}
\begin{pmatrix}
x_{kn}\\
\vdots \\
x_{kn} \\
\vdots \\
x_{k(n-1)}\\
\vdots \\
x_{k(n-1)} \\
\end{pmatrix} \nonumber \\
& = &
\bm X^T
\mathrm{diag}\Big\{d_{1k}, \cdots, d_{nk}, \cdots, d_{1k}, \cdots, d_{nk}\Big\}
\Big(\bm e^T_{(n-1)n+k}; \cdots; \bm e^T_{(n-1)n+k}; \nonumber  \\
&&\bm e^T_{k}; \cdots; \bm e^T_{k}; \cdots; \bm e^T_{(n-2)n+k}; \cdots; \bm e^T_{(n-2)n+k}\Big)
\bm X.
\end{eqnarray}
In a similar way,
\begin{eqnarray*}
\sum_{i,j}x_{ij}d_{ik}x_{k(j+1)} & = & \bm X^T
\mathrm{diag}\Big\{d_{1k}, \cdots, d_{nk}, \cdots, d_{1k}, \cdots, d_{nk}\Big\}\nonumber  \\
& &\Big(\bm e^T_{n+k}; \cdots; \bm e^T_{n+k};
\bm e^T_{2n+k}; \cdots; \bm e^T_{2n+k}; \cdots; \bm e^T_{k}; \cdots; \bm e^T_{k}\Big)
\bm X.
\end{eqnarray*}
It is easy to rewrite the constraints in vector forms as given in (\ref{eq4}) and (\ref{eq5}).
This completes the proof. \qed
\end{proof}

Without loss of generality, suppose that $x_{11} = 1$, then we have
\[
x_{1,j} = 0,  x_{i,1} = 0, \;\; \forall \;i,j = 2, \cdots, n.
\]

Let define $\bm Y = [x_{22}, x_{32}, \cdots, x_{n2}, \cdots, x_{2n}, x_{3n}, \cdots, x_{nn}]^T \in \mathbb{R}^{(n-1)^2 \times (n-1)^2}$ and rearrange $\bm X$ in such a way that
\begin{eqnarray*}
\hat{\bm X} &=& (\bm X_1; \bm Y) \\
\hat{\bm X} &=& \bm X(\bm {id})
\end{eqnarray*}
where, $\bm X_1 = [x_{11}, x_{21},\cdots, x_{n1}, x_{12}, x_{13}, \cdots, x_{1n}]^T$, and $\bm {id}$ is a unique index vector to establish the relationship between $\bm X$
and $\hat{\bm X}$.

To make sure that the objective function value is constant, the matrix $A$ should be rearranged through the following procedures
\begin{eqnarray*}
\hat{A} &=& A, \\
\hat{A}(n\!+\!2:n^2\!-\!n\!+\!1,:) &=& A(\bm {id},:), \\
\hat{A}(:,n\!+\!2:n^2\!-\!n\!+\!1) &=& A(:,\bm {id}).
\end{eqnarray*}

Let the rearranged matrix $\hat{A}$ be partitioned into
\begin{eqnarray*}
\hat{A} =
\begin{pmatrix}
\hat{A}_{11} & \hat{A}_{12} \\
\hat{A}_{21} & \hat{A}_{22}
\end{pmatrix}
\end{eqnarray*}
where, $\hat{A}_{11} \in \mathbb{R}^{(2n\!-\!1) \times (2n\!-\!1)}$, $\hat{A}_{12} \in \mathbb{R}^{(2n\!-\!1) \times (n\!-\!1)^2}$, $\hat{A}_{21} \in \mathbb{R}^{(n\!-\!1)^2 \times (2n\!-\!1)}$
and $\hat{A}_{22} \in \mathbb{R}^{(n\!-\!1)^2 \times (n\!-\!1)^2}$, then we have
\begin{eqnarray*}
\frac{1}{2} \hat{\bm X}^T \hat{A} \hat{\bm X} &=& \frac{1}{2} (\bm X_1; \bm Y)^T
\begin{pmatrix}
\hat{A}_{11} & \hat{A}_{12} \\
\hat{A}_{21} & \hat{A}_{22}
\end{pmatrix}
(\bm X_1; \bm Y) \\
&=& \frac{1}{2} \bm Y^T \hat{A}_{22} \bm Y + \frac{1}{2}(\bm Y^T \hat{A}_{21} \bm X_1 + \bm X_1^T \hat{A}_{12} \bm Y) + \frac{1}{2} \bm X^T \hat{A}_{11} \bm X.
\end{eqnarray*}
As a result, the TSP problem can be reduced to
\begin{eqnarray}
(\mathcal{P}_r) ~~\min && f(\bm Y) = \frac{1}{2} \bm Y^T A_r \bm Y - \bm b_r^T \bm Y\nonumber \\
\mathrm{s.t.} && E_r \bm Y = \bm e \nonumber \\
&& \bm Y \circ \bm Y - \bm Y = \bm 0
\end{eqnarray}
where, $A_r = \hat{A}_{22}$, $\bm b_r = -\frac{1}{2}(\hat{A}_{21}\bm X_1 + \hat{A}_{12}^T\bm X_1)$ and $E_r = (C_r;D_r)$, here,
$C_r \in \mathbb{R}^{(n\!-\!1) \times (n\!-\!1)^2}$,  $D_r \in \mathbb{R}^{(n\!-\!2) \times (n\!-\!1)^2}$ (the last row is deleted) and
\begin{eqnarray*}
C_r = && \begin{pmatrix}
1 & \cdots & 1 & 0 & \cdots & 0 & 0 & \cdots & 0 \\
0 & \cdots & 0 & 1 & \cdots & 1 & 0 & \cdots & 0 \\
\vdots & \ddots & \vdots & \vdots & \ddots & \vdots & \vdots & \ddots & \vdots \\
0 & \cdots & 0 & 0 & \cdots & 0 & 1 & \cdots & 1 \\
\end{pmatrix},
\;
D_r = \begin{pmatrix}
1 & 0 & \cdots & 0 & \cdots & \cdots & 1 & 0 & \cdots & 0 \\
0 & 1 & \cdots & 0 & \cdots & \cdots & 0 & 1 & \cdots & 0 \\
\vdots & \vdots & \ddots & \vdots &  \vdots & \vdots & \vdots & \vdots & \ddots & \vdots \\
\xout{0} & \xout{0} & \xout{\cdots} & \xout{1} & \xout{\cdots} & \xout{\cdots} & \xout{0} & \xout{1} & \xout{\cdots} & \xout{1} \\
\end{pmatrix}
\end{eqnarray*}
\section{Classic Lagrangian method}
By introducing the Lagrange multiplier $\bm \lambda$, $\bm \mu$ associated with constraint $E_r \bm Y - \bm e = 0$, and $\bm Y \circ \bm Y - \bm Y = \bm 0$, respectively, the Lagrangian $L:\mathbb{R}^{(n\!-\!1)^2} \times \mathbb{R}^{2n\!-\!3} \times \mathbb{R}^{(n\!-\!1)^2}\rightarrow \mathbb{R}$ can be defined as
\begin{eqnarray}
L(\bm Y, \bm \lambda,\bm \mu) &=& \frac{1}{2} \bm Y^T A_r \bm Y - \bm b_r^T \bm Y + \bm \lambda^T (E_r \bm Y - \bm e) + \frac{1}{2} \bm \mu^T (\bm Y \circ \bm Y - \bm Y)\nonumber \\
&=&  \frac{1}{2} \bm Y^T A_r(\bm \lambda,\bm \mu) \bm Y - \bm Y^T \bm b_r(\bm \lambda,\bm \mu) - \bm \lambda^T  \bm e
\end{eqnarray}
where,
\[
A_r(\bm \lambda,\bm \mu) = A_r + \mathrm{diag}\{\bm \mu\}, \bm b_r(\bm \lambda,\bm \mu) = \bm b_r + \frac{1}{2}\bm \mu - E_r^T \bm \lambda
\]

The Lagrangian dual function can be obtained by
\begin{eqnarray}
g(\bm \lambda, \bm \mu) = \inf_{\bm Y} L(\bm Y, \bm \lambda, \bm \mu).
\end{eqnarray}

Let define the following dual feasible space
\begin{eqnarray}
\mathcal{S}^{+} = \{(\bm \lambda, \bm \mu)\in \mathbb{R}^{2n\!-\!3} \times \mathbb{R}^{(n\!-\!1)^2} | A_r(\bm \lambda,\bm \mu) \succ 0\},
\end{eqnarray}
then the Lagrangian dual function can be written explicitly as
\begin{eqnarray}
g(\bm \lambda,\bm \mu) = -\frac{1}{2} \bm b_r^T(\bm \lambda,\bm \mu) A_r^{-1}(\bm \lambda,\bm \mu) \bm b_r(\bm \lambda,\bm \mu) - \bm \lambda^T  \bm e, \;\; \mathrm{s.t.}\;\; (\bm \lambda,\bm \mu) \in  \mathcal{S}^{+}
\end{eqnarray}
associated with the Lagrangian equation
\begin{eqnarray}
A_r(\bm \lambda,\bm \mu) \bm Y = \bm b_r(\bm \lambda,\bm \mu).
\end{eqnarray}
Finally, the Lagrangian dual problem can be obtained as
\begin{eqnarray}
\max_{(\bm \lambda, \bm \mu) \in  \mathcal{S}^{+}} \Big\{g(\bm \lambda,\bm \mu) = -\frac{1}{2} \bm b_r^T(\bm \lambda,\bm \mu) A_r^{-1}(\bm \lambda,\bm \mu) \bm b_r(\bm \lambda,\bm \mu) - \bm \lambda^T  \bm e\Big\}.
\end{eqnarray}

\begin{theorem} \label{the2}
If $(\bar{\bm \lambda}, \bar{\bm \mu})$ is a critical point of the Lagrangian dual function and $(\bar{\bm \lambda}, \bar{\bm \mu}) \in \mathcal{S}^{+}$, then the corresponding
$\bar{\bm Y} = A_r^{-1}(\bar{\bm \lambda},\bar{\bm \mu}) \bm b_r(\bar{\bm \lambda},\bar{\bm \mu}))$ is a global solution to the reduced TSP problem ($\mathcal{P}_r$).
\end{theorem}
\begin{proof}
The proof is trivial and is omitted here.
\end{proof}
\section{Inverse problem}
The inverse problem can be simplified as follows
\begin{eqnarray} \label{inverseprob}
&\mathrm{find}& \;\;\;\; A_r, \bm b_r, \bm Y, \bm \lambda, \bm \mu \nonumber \\
&\mathrm{s.t.}& \;\;\;(A_r + \mathrm{diag}\{\bm \mu\}) \bm Y = \bm b_r + \frac{1}{2}\bm \mu - E_r^T \bm \lambda  \nonumber \\
&& \;\;\;A_r + \mathrm{diag}\{\bm \mu\} \succ 0 \nonumber \\
&& \;\;\; E_r \bm Y = \bm e \nonumber \\
&& \;\;\; \bm Y \circ \bm Y = \bm Y
\end{eqnarray}

This a feasibility problem in optimization. To solve such an inverse problem, some degree of freedom should be given in advance. For instance,
we can suppose that $\bm Y$ is a freely random "true" solution. 

\section{Numerical experiments}
Now, let consider a 4-cities TSP problem, whose distance matrix is given as follows
\begin{eqnarray*}
\bm d =
\begin{pmatrix}
         0   & d_{12}  &  d_{13}  &  d_{14}\\
    d_{21}   &      0  &  d_{23}  &  d_{24}\\
    d_{31}   & d_{32}  &       0  &  d_{34}\\
    d_{41}   & d_{42}  &  d_{43}  &       0\\
\end{pmatrix} =
\begin{pmatrix}
0       & \bm d_1^T \\
\bm d_1 & \bm d_2
\end{pmatrix}
\end{eqnarray*}
where, $\bm d_1 = (d_{12},d_{13}, d_{14})^T$ and $\bm d_2$ is the remainder,
then
\begin{eqnarray*}
A_r =
\begin{pmatrix}
\bm 0   & \bm d_2 & \bm 0\\
\bm d_2 & \bm 0   & \bm d_2\\
\bm 0   & \bm d_2 & \bm 0\\
\end{pmatrix},
\bm b_r =
\begin{pmatrix}
- \bm d_1\\
\bm 0 \\
- \bm d_1\\
\end{pmatrix},
E_r =
\begin{pmatrix}
     1  &   1  &   1  &   0  &   0  &   0  &   0  &   0  &   0\\
     0  &   0  &   0  &   1  &   1  &   1  &   0  &   0  &  0\\
     0  &   0  &   0  &   0  &   0  &   0  &   1  &   1  &   1\\
     1  &   0  &   0  &   1  &   0  &   0  &   1  &   0  &   0\\
     0  &   1  &   0  &   0  &   1  &   0  &   0  &   1  &   0\\
\end{pmatrix}.
\end{eqnarray*}
Suppose that $\bar{\bm Y} = (1,0,0,0,1,0,0,0,1)^T$, to design such a TSP problem, $\bm d$ should satisfy
\begin{eqnarray}
(A_r + \mathrm{diag}\{\bm \mu\}) \bar{\bm Y} &=& \bm b_r + \frac{1}{2}\bm \mu - E_r^T \bm \lambda  \nonumber \\
A_r + \mathrm{diag}\{\bm \mu\} &\succ& 0, 
\end{eqnarray}
which corresponds to the former two conditions in (\ref{inverseprob}), and 
\begin{eqnarray}
d_{12} + d_{23} + d_{34} + d_{41} &<& d_{13} + d_{32} + d_{24} + d_{41} \nonumber \\
d_{12} + d_{23} + d_{34} + d_{41} &<& d_{13} + d_{34} + d_{42} + d_{21},  
\end{eqnarray}
which is to guarantee that $\bar{\bm Y}$ is a best solution, and 
\begin{eqnarray}
d_{ij} &>& 0 \nonumber \\
d_{ij} &=& d_{ji} \nonumber \\
d_{ij} &\leq& d_{ik} + d_{kj}, i \neq j \neq k, 
\end{eqnarray}
which is to guarantee that a Euclidean distance matrix \cite{EDM} is satisfied, 
or more specifically
\begin{eqnarray*}
\begin{pmatrix}
d_{23} + \mu_1\\
0 \\
d_{43}\\
d_{24}\\
d_{32} + d_{34} + \mu_{5}\\
d_{42}\\
d_{23}\\
0 \\
d_{43} + \mu_9\\
\end{pmatrix} =
\begin{pmatrix}
-d_{12}\\
-d_{13}\\
-d_{14}\\
0\\
0\\
0\\
-d_{12}\\
-d_{13}\\
-d_{14}\\
\end{pmatrix}
+
\frac{1}{2}
\begin{pmatrix}
\mu_{1}\\
\mu_{2}\\
\mu_{3}\\
\mu_{4}\\
\mu_{5}\\
\mu_{6}\\
\mu_{7}\\
\mu_{8}\\
\mu_{9}\\
\end{pmatrix}
-
\begin{pmatrix}
\lambda_{1} + \lambda_{4}\\
\lambda_{1} + \lambda_{5}\\
\lambda_{1}\\
\lambda_{2} + \lambda_{4}\\
\lambda_{2} + \lambda_{5}\\
\lambda_{2}\\
\lambda_{3} + \lambda_{4}\\
\lambda_{3} + \lambda_{5}\\
\lambda_{3} \\
\end{pmatrix}
\end{eqnarray*}
\begin{eqnarray*}
\begin{pmatrix}
\mu_1 & 0 & 0 & 0 & d_{23} & d_{24} & 0 & 0 & 0 \\
0 & \mu_2 & 0 & d_{32} & 0 & d_{34} & 0 & 0 & 0\\
0 & 0 & \mu_{3} & d_{42} & d_{43} & 0 & 0 & 0 & 0\\
0 & d_{23} & d_{24} & \mu_{4} & 0 & 0 & 0 & d_{23} & d_{24}\\
d_{32} & 0 & d_{34} & 0 & \mu_{5} & 0 & d_{32} & 0 & d_{34} \\
d_{42} & d_{43} & 0 & 0 & 0 & \mu_{6} & d_{42} & d_{43} & 0 \\
0 & 0 & 0 & 0 & d_{23} & d_{24} & \mu_{7} & 0 & 0    \\
0 & 0 & 0 & d_{32} & 0 & d_{34} & 0 & \mu_{8} & 0  \\
0 & 0 & 0 & d_{42} & d_{43} & 0 & 0 & 0 & \mu_{9}  \\
\end{pmatrix}
\succ 0
\end{eqnarray*}
\begin{eqnarray*}
d_{12} + d_{23} + d_{34} + d_{41} &<& d_{13} + d_{32} + d_{24} + d_{41} \nonumber \\
d_{12} + d_{23} + d_{34} + d_{41} &<& d_{13} + d_{34} + d_{42} + d_{21} \nonumber \\
d_{ij} &>& 0 \nonumber \\
d_{ij} &=& d_{ji} \nonumber \\
d_{ij} &\leq& d_{ik} + d_{kj}, i \neq j \neq k
\end{eqnarray*}

We try to solve the above feasible problem (involving all variables $\bm d, \bm \lambda,  \bm \mu$) numerically by YALMIP \cite{lofberg2004yalmip}, but no feasible solutions can be obtained, which may be an indication that
the traveling salesman problem can not be solved by classic Lagrangian.



\end{document}